\documentclass[10pt,conference]{IEEEtran}

\usepackage[dvips]{graphicx}
\usepackage{times}
\usepackage{cite}
\usepackage{amsmath}
\usepackage{array}
\usepackage{amssymb}

\usepackage{bbm}

\usepackage{stfloats}
\usepackage{rotating,threeparttable,booktabs}
\usepackage{bm}
\usepackage{dcolumn,booktabs}
\usepackage{multirow}
\usepackage{graphicx}
\usepackage{subfigure}
\usepackage{color}		
\usepackage{amsmath,amsthm,amssymb,amsfonts,mathrsfs}

\makeatletter
\thm@headfont{\sc}
\makeatother
\newtheorem{theorem}{Theorem}

\usepackage{capt-of} 
\usepackage{url}
\usepackage[noend]{algpseudocode}
\usepackage{algorithmicx,algorithm}
\usepackage{CJK}

\begin{document}
\title{Passive Beamforming Design and Channel Estimation for IRS Communication System with Few-Bit ADCs}
\author{\IEEEauthorblockN{Jingnan Li, Rui Wang and Erwu Liu}
	\IEEEauthorblockA{
		College of Electronics and Information Engineering, Tongji University, Shanghai, China \\
					Emails: 1930712@tongji.edu.cn, ruiwang@tongji.edu.cn, erwu.liu@ieee.org\\				
		 }}
\maketitle

\begin{abstract}
	Utilizing intelligent reflecting surface (IRS) was proven to be efficient in improving the energy efficiency for wireless networks. In this paper, we investigate the passive beamforming and channel estimation for IRS assisted wireless communications with low-resolution analog-to-digital converters (ADCs) at the receiver. We derive the approximate achievable rate by using the Bussgang theorem. Based on the derived analytical achievable rate expression, we maximize the achievable rate by using semidefinite programming (SDP), branch-and-bound (BB), and gradient-based approaches. A maximum likelihood (ML) estimator is then proposed for channel estimation by considering the $\mathrm{1}$-bit quantization ADC. Numerical result shows that the proposed beamforming design and channel estimation method significantly outperforms the existing methods.
\end{abstract}


\section{Introduction}
	In recent years, there has been a rapid increase in the spectrum utilization efficiency of wireless networks, thanks to various technological advances such as massive multiple-input multiple-output (MIMO) and millimeter wave communications. However, the high cost of the network in terms of hardware and energy consumption remains a serious challenge in practical implementation \cite{SZhang_Qingqing_2017}. In the sixth generation (6G) system, the peak data rate is expected to reach 100 times that of the fifth generation (5G), i.e., at least $1 \mathrm{T}$ bit/s \cite{ZZhang6G_2019}. The high speed and large amount of data transmission will inevitably consume a lot of energy. To address this challenge, intelligent reflecting surface (IRS) has recently been proposed and attracted great attentions from academia and industry due to the ability of alleviating the dilemma of high energy consumption. IRS is a passive component that can reflect electromagnetic waves and send signals directly without processing, which is completely different from the traditional strategy. Therefore, IRS perfectly meets the requirements of green communication in 6G system. The IRS controls the communication channel intelligently and can be conveniently arranged on the outer wall of the building to solve the problem of high frequency communication. Considering the advantage of the IRS, the 6G white paper also regards it as a key technology \cite{WhitePaper6G-2020}.

	Prior works on IRS-assisted communication system can be found in \cite{QingQing-GLOBECOM2018,QingQing-TWC2019,YanWenjing_IWC2020}. Passive beamforming design and channel estimation have been recognized as two important techniques for the IRS-assisted wireless communication systems. Q. Wu and R. Zhang optimized the reflecting beamforming direction in IRS assisted massive MIMO and multiuser-input single-output (MISO) system by applying semidefinite relaxation (SDR) and alternating optimization techniques \cite{QingQing-GLOBECOM2018,QingQing-TWC2019}. W. Yan, \textit{et al.}, optimized the phase of elements in IRS to maximize the signal-to-noise ratio (SNR) for an IRS-assisted single-input multiuser-output
	(SIMO) system assuming that each element of the IRS is independently opened with a preset \cite{YanWenjing_IWC2020}. Channel estimation of the IRS-assisted system was investigated in \cite{YuanX-WCL2019,ZWang-TWC2020}. Specifically, X. Yuan, \textit{et al}., proposed a three-stage algorithm to estimate the cascaded channel in the IRS-assisted massive MIMO system \cite{YuanX-WCL2019}. Z. Wang, \textit{et al.}, proposed a novel three-phase framework to estimate channel in an IRS-assisted massive MIMO system \cite{ZWang-TWC2020}.
		
 	Although the studies of the IRS-assisted wireless communications have received wide attention, it is worth noting that exiting works assume receiver with infinite or high precision quantization. In practice, we usually apply finite precision quantization analog-to-digital converters (ADCs) to save the processing power at receivers. Therefore,  in this paper, we consider an IRS-assisted SIMO wireless system with few-bit ADCs, which is a novel exploration compared to the works in \cite{QingQing-GLOBECOM2018, QingQing-TWC2019,YanWenjing_IWC2020,YuanX-WCL2019, ZWang-TWC2020}.  	
 	We optimize the passive beamforming by designing the phase shifts at the IRS. In passive beamforming design, we formulate the problem of maximizing the lower bound of achievable rate. In particular, to simplify the computation at a low SNR regime, we approximate the problem and obtain a sub-optimal solution by applying the SDR technique. Moreover, the design problem is also transformed into a complex quadratic programming (CQP) problem and solved by the branch-and-bound (BB) algorithm. To optimize the passive beamforming for arbitrary SNR regime, we propose a method applying gradient descent. As for the channel estimation of the IRS system with few-bit ADC, a maximum likelihood (ML) estimator is proposed and the estimation error of the direct channel is taken into account when estimating the IRS channel link. In specific, the closed form expression of the direct channel estimation error is obtained by deducing the Cram{\'{e}}r-Rao Lower Bound (CRLB). Extensive simulation results are provided to show the superior performance of the proposed beamforming design and channel estimation method.
 		

\section{System Model}
As shown in Fig. 1, we consider a SIMO uplink wireless communication system where a single-antenna user communicates with a base station (BS) equipped with $M$ antennas, and each antenna of BS has a pair of ADCs. The IRS composed of $N$ passive reflecting elements is aided to assist the user in communicating with BS. In practice, there has a lot of sensors and each IRS is attached with a controller. The controller can appropriately adjusts the on/off state of each passive reflecting element according to the environmental data sent from sensors via a wired link. At the same time, the controller adjusts the phase shifts according to the channel state information (CSI) via a wireless link to improve the performance of the system. We only consider the signal that has been reflected once and ignore the signals reflected by IRS for two or more times due to serious path loss. We assume that all the channel links are quasi-static and flat-fading.
	
	In Fig. 1, $\mathbf{h}_{d}\in\mathbb{C}^{M\times 1}$,  $\mathbf{h}_{r}\in\mathbb{C}^{N\times 1}$,  $\mathbf{\textbf{\emph{G}}}\in\mathbb{C}^{M\times N}$ denote the baseband channels from user to BS, from user to IRS, from IRS to BS, respectively. The diagonal phase-shift matrix of the IRS is denoted by $\pmb\Theta=\mathrm{diag} (\beta_1 e^{j\theta_1},\beta_2 e^{j\theta_2},\ldots,\beta_N e^{j\theta_N})$, where $\theta_n\in[0, 2\pi)$ and $\beta_n\in[0, 1]$ are the phase shift and amplitude coefficient of the $n$-th passive reflecting element. Without loss of generality, we set $\beta_n=1,\forall{n}$ by considering the practical situation that each element of IRS is expected to make contributions for optimizing the system performance \cite{QingQing-TWC2019}.

	\begin{figure}[t]
		\centering
			\vspace{0.5cm}
		\includegraphics[scale=0.2]{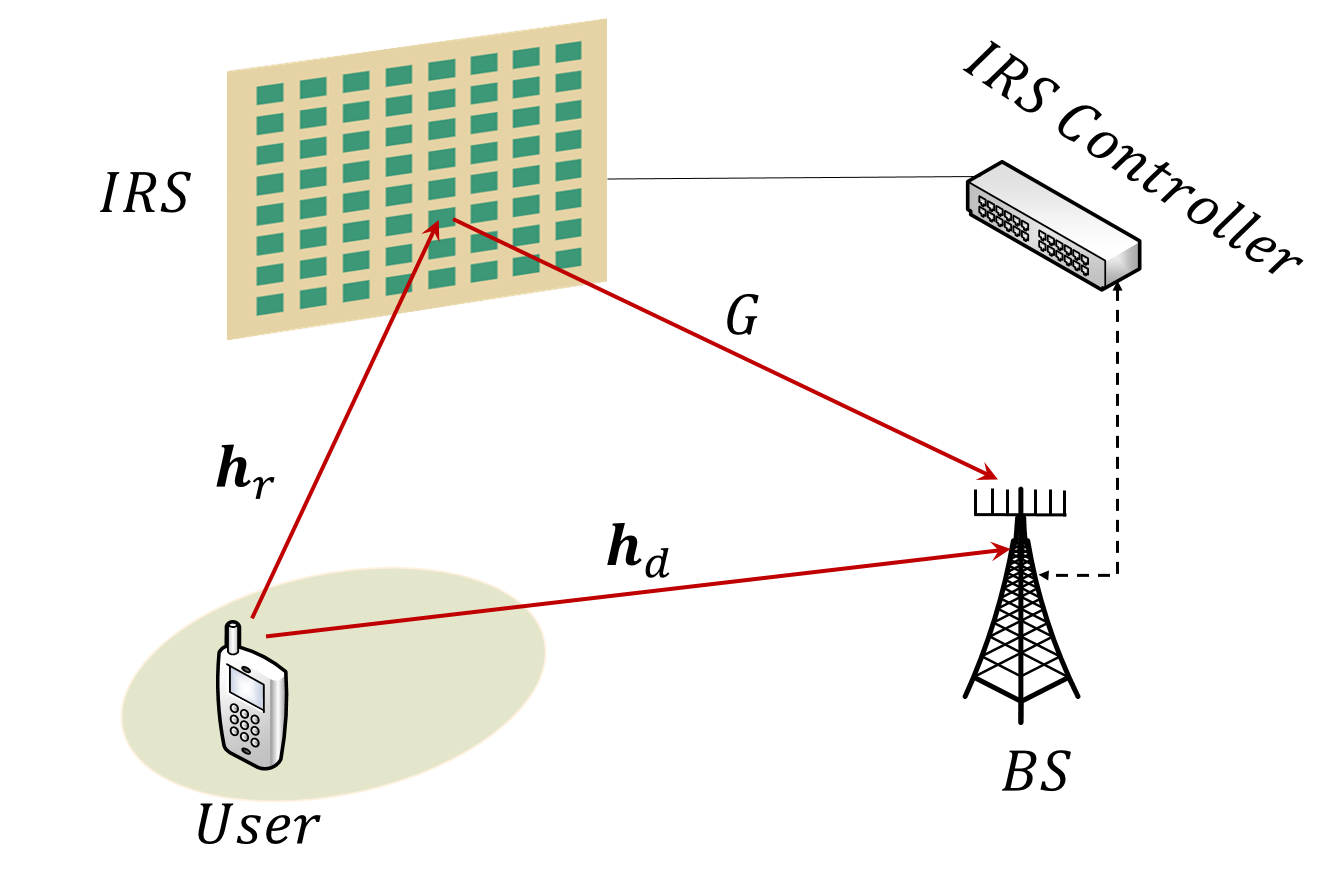}
		\vspace{-0.1cm}
		\caption{An IRS-assisted SIMO wireless communication system with few-bit ADCs at BS.} \label{Fig_mRC}
		\centering
		\vspace{-0.3cm}
	\end{figure}

	The received baseband signal consists of two terms, one is directly transmitted from the user and the other is reflected by the IRS. Therefore, the signal received at BS is given by
		\begin{equation}\label{System-1}
			\begin{split}
				\mathbf{y}=(\mathbf{G}\pmb\Theta\mathbf{h}_{r}+\mathbf{h}_{d})\text{\emph{x}}+\mathbf{w},\\
			\end{split}
		\end{equation}
	where $\mathbf{y}\in\mathbb{C}^{M\times 1}$ denotes the received signal, $\text{\emph{x}}$ is the transmit data of the user, $\mathbf{w}\in\mathbb{C}^{M\times 1}$ is an additive white Gaussian noise (AWGN) with the elements drown from $\mathcal{CN}(0,\sigma_{\mathit{w}}^{2})$ independently. The received signal $\mathbf{y}$ after quantization is denoted by	
		\begin{equation}\label{System-2}
			\begin{split}
				\mathbf{r}=\mathcal{Q}(\mathbf{y})
									=\mathcal{Q}((\mathbf{G}\pmb\Theta\mathbf{h}_{r}+\mathbf{h}_{d})\text{\emph{x}}+\mathbf{w}),\\
			\end{split}
		\end{equation}
where $\mathcal{Q}$ is the quantization operation of the ADC.

\section{Passive Beamforming Design}
By assuming that the CSI is perfectly known at both transmitter and receiver, in this section we focus on designing the passive beamforming with an aim to maximize the lower bound of achievable rate through adjusting the phase shift of each element in the IRS. To this end, we first derive the achievable rate of the considered system by using Bussgang theorem. The received signal after quantization is
		\begin{equation}\label{System-3}
			\begin{split}
				\mathbf{r}&=\mathbf{Fy}+\mathbf{e},\\
			\end{split}
		\end{equation}
	where $\mathbf{F}$ is a linear operation matrix which can be obtained from minimum mean square error of $\mathbf{r}$ from $\mathbf{y}$, and $\mathbf{e}$ is the quantization error. The expression of the achievable rate between receiver and transmitter can be approximated as \cite{Mezghani_IWC2012}
		\begin{equation}\label{System-4}
			\begin{split}
				I=&\mathrm{log}_2|\mathbf{\textbf{I}}+(1-\rho_q)((1-\rho_q)\mathbf{R}_\mathbf{ww}+\rho_q \mathrm{diag}(\mathbf{R}_\mathbf{yy}))^{-1}\\
				&		\times\mathbf{h}\sigma_{\mathit{x}}^{2}\mathbf{h}^H|,\\
			\end{split}
		\end{equation}
	 where $\rho_q$ is a distortion factor depending on the type of quantizer and the number of quantization bits at the BS \cite{Mezghani_IWC2012}, $\mathbf{R}_{\mathbf{ww}}=\mathbf{E}[\mathbf{ww^H}]$ is the covariance matrix of $\mathbf{w}$, $\sigma_{\mathit{x}}^{2}=\mathbf{E}[\textit{xx}^H]$ is the power of transmit signal, $\mathbf{R}_\mathbf{yy}=		\mathbf{R}_\mathbf{ww}+\sigma_{\mathit{x}}^{2}\mathbf{h}\mathbf{h}^H$, $\mathbf{h}=\mathbf{G}\pmb\Theta\mathbf{h}_{r}+\mathbf{h}_{d}\in\mathbb{C}^{M\times 1}$. To maximize the achievable rate, we consider two specific scenarios, as listed below.
	
	\subsection{Beamforming design at Low SNR}	
	At low SNR regime, the approximate lower bound of the mutual information is approximated by
		\begin{equation}\label{System-5}
			\begin{split}
				I_{low SNR}\approx				
				\mathrm{tr} \bigr( \sigma_{\mathit{x}}^{2}(1-\rho_q)\mathbf{R}_\mathbf{ww}^{-1}\mathbf{h}\mathbf{h}^H\bigr).\\						
			\end{split}
		\end{equation}
	With \eqref{System-5}, the optimization of the phase shift at the IRS is given by
		\begin{equation}\label{System-P} \nonumber
			\begin{split}
				(P):  \max_{\bf\Theta}\quad&\mathrm{tr} (\mathbf{R}_\mathbf{ww}^{-1}\mathbf{h}\mathbf{h}^H)\\
				s.t \quad&\theta_i\in[0,2\pi), \quad\emph{i}=1,2,\ldots,N.
			\end{split}
		\end{equation}
	Plugging $\mathbf{h}=\mathbf{G}\pmb\Theta\mathbf{h}_{r}+\mathbf{h}_{d}$ into the objective function of P, we have
			\begin{equation}\label{System-trace}
				\begin{split}
					\mathrm{tr}\Big(\mathbf{R}_\mathbf{ww}^{-1}\mathbf{h}&\mathbf{h}^H\Big)
					=2\mathrm{Re}\left(  \mathbf{h}_{d}^{H}\mathbf{R}_\mathbf{ww}^{-1}\mathbf{G}\mathfrak{R} \mathbf{u} \right)	\\
					&+\mathbf{u}^H\mathfrak{R}^{H}\mathbf{G}^{H}\mathbf{R}_\mathbf{ww}^{-1}\mathbf{G}\mathfrak{R}\mathbf{u}
					+\mathrm{tr}\big(\mathbf{h}_{d}\mathbf{h}_{d}^{H}\mathbf{R}_\mathbf{ww}^{-1}\big),\\
				\end{split}
			\end{equation}
		where $\mathfrak{R}=\mathrm{diag}(\mathbf{h}_{r})$, $\mathbf{u}=(u_1,\ldots,u_N)^T$, and $u_n=e^{j\theta_n}$, $\forall n$. With \eqref{System-trace}, the problem $P$ can be simplified to	
			\begin{equation}\label{Eqn_1} \nonumber
				\begin{split}
					(P1):\max_{\bf{u}}\quad& 2\mathrm{Re}\left(  \mathbf{h}_{d}^{H}\mathbf{R}_\mathbf{ww}^{-1}\mathbf{G}\mathfrak{R} \mathbf{u} \right)
					+\mathbf{u}^H\mathfrak{R}^{H}\mathbf{G}^{H}\mathbf{R}_\mathbf{ww}^{-1}\mathbf{G}\mathfrak{R}\mathbf{u}\\
					s.t\quad  &|u_i|=1, \quad\emph{i}=1,2,\ldots,N.
				\end{split}
			\end{equation}
			
	To solve the non-convex optimization problem  $P$1, we apply SDR technique and reformulate it into a semidefinite programming (SDP) problem or a CQP problem which can be solved by branch-and-bound algorithm efficiently.

	\subsubsection{SDP}

		Problem $P1$ is a non-convex quadratically constrained quadratic program (QCQP), it can be approximated as a SDP problem following \cite{QingQing-TWC2019} and \cite{YanWenjing_IWC2020}. We import an instrumental variable $t$ to reformulate problem in $P1$ as a homogeneous QCQP		
			\begin{equation} \label{Eqn_2}
   				\begin{split}
					\max_\mathbf{\bar{\mathbf{u}}}\quad & \bar{\mathbf{u}}^H(\mathbf{R}+\mathbf{V})\bar{\mathbf{u}} \\
					s.t. \quad& |u_n|=1, \forall{n}=1,\ldots ,N+1,
				\end{split}
			\end{equation}
		where $\bar{\mathbf{u}}= \begin{bmatrix} \mathbf{u} \\ t \end{bmatrix}$,  $\mathbf{R}=\begin{bmatrix}
		\mathfrak{R}^H\mathbf{G}^H\mathbf{R}_\mathbf{ww}^{-1}\mathbf{G}\mathfrak{R} & \mathfrak{R}^H\mathbf{G}^H\mathbf{R}_\mathbf{ww}^{-1}\mathbf{h}_{d} \\
		{\mathbf{h}_{d}}^H\mathbf{R}_\mathbf{ww}^{-1}\mathbf{G}\mathfrak{R}&0
		\end{bmatrix}$, $\mathbf{\emph{V}}=\begin{bmatrix}
		 \mathbf{\emph{O}}_{N\times N}& \mathbf{\textbf{\emph{0}}}_{N\times 1} \\
		{\mathbf{h}_{d}}^H\mathbf{R}_\mathbf{ww}^{-1}{diag}(\mathbf{h}_{r})&0
		\end{bmatrix}$. Note that $\bar{\mathbf{u}}^H(\mathbf{R}+\mathbf{V})\bar{\mathbf{u}}=\mathrm{tr}((\mathbf{R}+\mathbf{V})\mathbf{U})$, where $\mathbf{U}=\bar{\mathbf{u}}\bar{\mathbf{u}}^H$ is a positive semidefinite matrix with $\mathrm rank({\mathbf{U}})=1$. By relaxing the rank constraint, problem $P1$ is transformed into
			\begin{equation} \label{Eqn_3}
				\begin{split}
					\max_{\mathbf{U}} \quad& \mathrm{tr}\left( \left( \mathbf{R}+\mathbf{V}\right) \mathbf{U}\right) \\
					s.t. \quad&\mathbf{U}\succeq{0},\mathbf{U}_{n,n}=1, \forall{n}=1,\ldots ,N+1
				\end{split}
			\end{equation}
		which is a SDP problem, we can solve this problem by using convex optimization solvers such as CVX \cite{CVX}. Generally, the optimal $\mathbf{U}$ solved from problem \eqref{Eqn_3} may not be rank-one. To get the sub-optimal solution from $\mathbf{U}$ with $rank(\mathbf{U})\neq1$, we first take a eigenvalue decomposition of $\mathbf{U}$ as $\mathbf{U}=\mathbf{T}{\bf \Sigma}\mathbf{T}^H$, $\mathbf{T}$ is a Unitary matrix with size of $(N+1)\times(N+1)$, and ${\bf \Sigma}$ is a diagonal matrix. A sub-optimal solution can be obtained as $\bar{\mathbf{u}}= \mathbf{T}{\bf \Sigma}^{1/2}\pmb{\gamma}$, where $\pmb{\gamma}\in\mathbb{C}^{(N+1)\times 1}$ is a vector with each of its elements randomly drown from the circularly symmetric complex Gaussian (CSCG) distribution $\mathcal{CN}(0,1)$. We choose one $\pmb{\gamma}$ which attains the maximum value in \eqref{Eqn_3} from all likely vectors. Finally, the solution is given by $\mathbf{u}=\left(e^{j\mathrm{arg}(\frac{\bar u_1}{\bar u_{(N+1)} })},\cdots,e^{j\mathrm{arg}(\frac{\bar u_N}{\bar u_{(N+1)}})} \right)$.
		
	\subsubsection{CQP}					
		The result obtained by solving SDP may not be optimal if $rank(\mathbf{U})$ is not one. Next we convert the problem $P1$ into the form of CQP, and use branch-and-bound algorithm to find an approximately optimal solution \cite{ChengLu_arXiv}. 	
Denote   $\mathbf{Q}=-2\left(\mathfrak{R}^{H}\mathbf{G}^{H}\mathbf{R}_\mathbf{ww}^{-1}\mathbf{G}\mathfrak{R}\right)\in\mathbb{C}^{N\times N}$, $\mathbf{c}^H=-2\mathbf{h}_{d}^{H}\mathbf{R}_\mathbf{ww}^{-1}\mathbf{G}\mathfrak{R}$, problem $P$ can be reformulated as
				\begin{equation}\label{system-P1} \nonumber
					\begin{split}
						(P2):\min_{\mathbf{u}} \quad& \frac{1}{2}\mathbf{u}^H\mathbf{Q}\mathbf{u}+\mathrm{Re}(\mathbf{c}^H\mathbf{u})\\
						s.t. \quad|&\mathbf{u}_i|=1, \mathrm{arg}(\mathbf{u}_i)\in[0,2\pi), \emph{i}=1,2,\ldots,N.
					\end{split}
				\end{equation}
To solve $P2$, we import an $N\times N$ complex matrix $\mathbf\Xi= \mathbf{u}\mathbf{u}^H$ and apply SDR technique to convert $P2$ into

%
%
				\begin{equation}\label{system-P3} \nonumber
					\begin{split}
						(P3):\min_{\mathbf{u}} \quad& \frac{1}{2}\mathrm{tr}(\mathbf{Q}\mathbf\Xi)+\mathrm{Re}(\mathbf{c}^H\mathbf{u})\\
						s.t. \quad|&|\mathbf{u}_i|\leqslant 1,\\
								&|\mathbf\Xi_{ii}|=1, \emph{i}=1,2,\ldots,N, \\
								&\mathbf\Xi \succeq \mathbf{u}\mathbf{u}^H,
					\end{split}
				\end{equation}
in which we relaxe $\mathbf\Xi= \mathbf{u}\mathbf{u}^H$ by $\mathbf\Xi \succeq \mathbf{u}\mathbf{u}^H$	and drops the constraints $\mathrm{arg}(\mathbf{u}_i)\in[0,2\pi), \emph{i}=1,2,\ldots,N$. Thus, problem $P3$ can be solved by branch-and-bound algorithm \cite{ChengLu_arXiv}.
			
	\subsection{Beamforming design for general SNR}\label{SDR+SDP}
For general SNR, we propose a gradient descent based method to optimize the passive IRS beamforming. The achievable rate in \eqref{System-4} can be rewritten as
	 \begin{equation}\label{System-55}
		 \begin{split}
			 I=&\mathrm{log}_2(1+(1-\rho_q)\mathbf{h}^H(\mathbf{R}_\mathbf{ww}+\rho_q \mathrm{diag}(\mathbf{h}\mathbf{h}^H))^{-1}\mathbf{h}),\\
		 \end{split}
	 \end{equation}
and the beamforming design problem is given by
	 \begin{equation}\label{System-P4} \nonumber
		 \begin{split}
			 (P4):  \max_{\pmb\theta}\quad & \Gamma(\pmb\theta)=\mathbf{h}^H(\mathbf{R}_\mathbf{ww}+\rho_q \mathrm{diag}(\mathbf{h}\mathbf{h}^H))^{-1}\mathbf{h} \\
			 s.t \quad&\theta_i\in[0,2\pi), \quad\emph{i}=1,2,\ldots,N,
		 \end{split}
	 \end{equation}
 where $\pmb\theta=[\theta_1, \theta_2,\ldots,\theta_N]^T$. Due to the nonconvexity of $P4$, we solve the problem by gradient descent method and update $\pmb\theta$ as
		\begin{equation}\label{System-6}
			\begin{split}
				\pmb\theta^{k+1}=\pmb\theta^k+\alpha^k \nabla \Gamma(\theta^k),\\
			\end{split}
		\end{equation}
	where $\pmb\theta^{k}$ is the vector updated at the $k$-th iteration, $\alpha^k$ is the step size used in the $k$-th iteration, and $\nabla \Gamma(\pmb\theta^k)$ is  gradient at $\pmb\theta^k$. We move $\pmb\theta^{k}$ among the steepest descent with a step of $-\alpha^k \nabla \Gamma(\pmb\theta^k)$. By denoting $\pmb\aleph=(\mathbf{R}_\mathbf{ww}+\rho_q \mathrm{diag}(\mathbf{h}\mathbf{h}^H))^{-1}$, i.e., $\pmb\aleph=\mathrm{diag}(\frac{1}{\sigma_{\mathit{w}}^{2}+\rho_q |h_1|^2},\ldots,\frac{1}{\sigma_{\mathit{w}}^{2}+\rho_q|h_M|^2})$, the $i$-th element of $\mathbf{h}$ is given by
		\begin{equation} \label{Eqn_4}
			\begin{split}
				h_i &= \sum\limits_{k=1}^{N}g_{ik} h_{k,r} e^{j\theta_k}+h_{k,d},
			\end{split}
	\end{equation}
	where $g_{ik}$ is the $i$-th row and the $k$-th column element of $\mathbf{G}$, $h_{k,r}$ and $h_{k,d}$ are the $k$-th element of $\mathbf{h}_r$ and $\mathbf{h}_d$ respectively. We have
			\begin{equation}\label{System-7}
				\begin{split}
					\partial \Gamma(\pmb\theta)= \mathbf{h}^H\pmb\aleph\mathbf{h}^H	=\frac{M}{\rho_q}-\frac{\sigma_{\mathit{w}}^{2}}{\rho_q}\sum\limits_{k=1}^{M}\frac{1}{\sigma_{\mathit{w}}^{2}+\rho_q|h_k|^2},
				\end{split}
			\end{equation}	
	and the gradient of $\theta_i$ is
			\begin{equation}\label{System-8}
				\begin{split}
					\frac{\partial \Gamma( \pmb\theta)}{\partial \theta_i}=&\sigma_{\mathit{w}}^{2}\sum\limits_{k=1}^{M} [ (\bar{h_k}g_{ki}s_ih_{i,r}e^{j(\theta_i+\frac{\pi}{2})}+{h_k}\overline{g_{ki}s_ih_{i,r}}e^{-j(\theta_i+\frac{\pi}{2})})\\
					&\times (\sigma_{\mathit{w}}^{2}+\rho_q|h_k|^2)^{-2} ]\\
					=&\sigma_{\mathit{w}}^{2}s_i\sum\limits_{k=1}^{M}\left[ \frac{2Re(\bar{h_k}g_{ki}s_ih_{i,r}e^{j(\theta_i+\frac{\pi}{2})})}{(\sigma_{\mathit{w}}^{2}+\rho_q|h_k|^2)}\right].
				\end{split}
			\end{equation}	
	Thus, we derive the gradient of vector $\pmb\theta$, i.e., $\nabla \Gamma(\pmb\theta)=(\frac{\partial \Gamma(\pmb\theta)}{\partial \theta_1},\ldots,\frac{\partial \Gamma(\pmb\theta)}{\partial \theta_N})^T$.

\section{Channel estimation}
 In the IRS assisted wireless communication system with few-bit ADCs, we assume that CSI is well known in previous section. But in reality, the CSI is obtained by conducting the estimation. How to accurately estimate the channel is very important for beamforming design. In this section, we focus on the estimation of the direct channel and the reflecting channel with $1$-bit ADCs.
	
\subsection{Direct channel estimation}
In Phase $\mathrm{\uppercase\expandafter{\romannumeral1}}$, we estimate the direct channel by turning all elements of IRS into state ``off". In the channel estimation stage, user sends a pilot sequence consisting of $\tau$ symbols
	\begin{equation}\label{System-ad1}
		\begin{split}
			\textbf{\emph{a}}_\mathrm{\uppercase\expandafter{\romannumeral1}} =[a_1,\ldots,a_\tau]^T.
		\end{split}
	\end{equation}
The signal received at BS during time slot $\tau$ is
	\begin{equation}\label{System-10}
		\begin{split}
			\textbf{Y}_{p_{\mathrm{\uppercase\expandafter{\romannumeral1}}}}=\mathbf{h}_{d}\textbf{\emph{a}}_\mathrm{\uppercase\expandafter{\romannumeral1}}^T+\textbf{W}_{p_{\mathrm{\uppercase\expandafter{\romannumeral1}}}},\\
		\end{split}
	\end{equation}
	where $\textbf{Y}_{p_{\mathrm{\uppercase\expandafter{\romannumeral1}}}}\in\mathbb{C}^{M\times \tau}$, and $\textbf{W}_{p_{\mathrm{\uppercase\expandafter{\romannumeral1}}}}\in\mathbb{C}^{M\times \tau}$ is an AWGN. Next we vectorize the pilot signal matrix received by the BS as
	\begin{equation}\label{System-110}
		\begin{split}
			\textbf{y}_{p_{\mathrm{\uppercase\expandafter{\romannumeral1}}}}&=\mathrm{vec}(\textbf{Y}_{p_{\mathrm{\uppercase\expandafter{\romannumeral1}}}})=(\textbf{\emph{a}}_\mathrm{\uppercase\expandafter{\romannumeral1}}\otimes\textbf{I}_{M})\mathbf{h}_{d}+\textbf{w}_{p_{\mathrm{\uppercase\expandafter{\romannumeral1}}}}=\textbf{A}_{p_{\mathrm{\uppercase\expandafter{\romannumeral1}}}}\mathbf{h}_{d}+\textbf{w}_{p_{\mathrm{\uppercase\expandafter{\romannumeral1}}}}.
		\end{split}
	\end{equation}
The complex signal $\textbf{y}_{p_{\mathrm{\uppercase\expandafter{\romannumeral1}}}}$ can be written into real domain
	\begin{equation}\label{System-11}
		\begin{split}
			\mathbf{y}_{R,p_{\mathrm{\uppercase\expandafter{\romannumeral1}}}}=\mathbf{A}_{R,p_{\mathrm{\uppercase\expandafter{\romannumeral1}}}}\mathbf{h}_{d_{R}}+\mathbf{w}_{R,p_{\mathrm{\uppercase\expandafter{\romannumeral1}}}},\\
		\end{split}
	\end{equation}
where
	\begin{align}\label{System-12}
		\mathbf{y}_{R,p_{\mathrm{\uppercase\expandafter{\romannumeral1}}}}&=\begin{bmatrix} \mathrm{Re}(\mathbf{y}^T_{p_{\mathrm{\uppercase\expandafter{\romannumeral1}}}}) & \mathrm{Im}(\mathbf{y}^T_{p_{\mathrm{\uppercase\expandafter{\romannumeral1}}}}) \end{bmatrix}^T\in\mathbb{R}^{2M\tau\times 1},
		\\
		\mathbf{A}_{R,p_{\mathrm{\uppercase\expandafter{\romannumeral1}}}}&=
		\begin{bmatrix}
			\mathrm{Re}( \mathbf{A}_{p_{\mathrm{\uppercase\expandafter{\romannumeral1}}}})&
			-\mathrm{Im}(\mathbf{A}_{p_{\mathrm{\uppercase\expandafter{\romannumeral1}}}}) \\
			\mathrm{Im}( \mathbf{A}_{p_{\mathrm{\uppercase\expandafter{\romannumeral1}}}}) &
			\mathrm{Re}( \mathbf{A}_{p_{\mathrm{\uppercase\expandafter{\romannumeral1}}}})
		\end{bmatrix}
		\in\mathbb{R}^{2M\tau\times 2M},
		\\
		\mathbf{h}_{d_{R}}&=
		\begin{bmatrix}
			\mathrm{Re}( \mathbf{h}^T_{d_{R}}) &  \mathrm{Im}( \mathbf{h}^T_{d_{R}})
		\end{bmatrix}^T\in\mathbb{R}^{2M\times 1},
		\\
		\mathbf{w}_{R,p_{\mathrm{\uppercase\expandafter{\romannumeral1}}}}&=
		\begin{bmatrix}
			\mathrm{Re}(\mathbf{w}^T_{p_{\mathrm{\uppercase\expandafter{\romannumeral1}}}}) &
			\mathrm{Im}(\mathbf{w}^T_{p_{\mathrm{\uppercase\expandafter{\romannumeral1}}}})
		\end{bmatrix}^T
		\in\mathbb{R}^{2M\tau\times 1}.
	\end{align}
After $\mathrm{1}$-bit quantization, the signal can be expressed as $\textbf{r}_{p_{\mathrm{\uppercase\expandafter{\romannumeral1}}}}=\mathcal{Q}(\textbf{y}_{p_{\mathrm{\uppercase\expandafter{\romannumeral1}}}})$,  and the $i$-th output of the  $\mathrm{1}$-bit ADC is
	\begin{equation}\label{System-15}	
		\begin{split}
		 	\textbf{r}_{R,{p_{\mathrm{\uppercase\expandafter{\romannumeral1}}}},i}= \mathrm{sgn}(\text{\emph{y}}_{R,p_{\mathrm{\uppercase\expandafter{\romannumeral1}}},i}),\\
		\end{split}
	\end{equation}	
where sgn$(\cdot)$ is the sign function
	\begin{equation}\label{System-sign}	
		\mathrm{sgn}(x)=
		\begin{cases}
			1 \quad &if \quad x\geq 0\\
			-1 \quad &if \quad x<0
		\end{cases}.
	\end{equation}
We define the $i$-th line of $\mathbf{A}_{R,p_{\mathrm{\uppercase\expandafter{\romannumeral1}}}}$ as $\textbf{\emph{a}}_{R,p_{\mathrm{\uppercase\expandafter{\romannumeral1}}},i}^T$, i.e., $\mathbf{A}_{R,p_{\mathrm{\uppercase\expandafter{\romannumeral1}}}}=[\textbf{\emph{a}}_{R,p_{\mathrm{\uppercase\expandafter{\romannumeral1}}},i},\dots,\textbf{\emph{a}}_{R,p_{\mathrm{\uppercase\expandafter{\romannumeral1}}},2M\tau}]^T$. Based on the definition of $\textbf{\emph{a}}_{R,p_{\mathrm{\uppercase\expandafter{\romannumeral1}}},i}$, the base station performs the sign-refinement which can be expressed as
	\begin{equation}\label{System-16}	
		\begin{split}
			\tilde{\textbf{\emph{a}}}_{R,p_{\mathrm{\uppercase\expandafter{\romannumeral1}}},i}= \textbf{r}_{R,{p_{\mathrm{\uppercase\expandafter{\romannumeral1}}}},i}\textbf{\emph{a}}_{R,p_{\mathrm{\uppercase\expandafter{\romannumeral1}}},i}.
		\\
		\end{split}
	\end{equation}	
%
The ML channel estimator is given as \cite{MoJianhua_TOC2016}
	\begin{equation}\label{System-17}	
		\begin{split}
			\breve{\mathbf{h}}_{d_{R,ML}}=\underset{\acute{\mathbf{h}}_{d_{R}}\in\mathbb{R}^{2M\times 1}}{\mathrm{argmax}}\overset{2M\tau}{\underset{i=1}{\sum}}\mathrm{ln}(\Phi(\sqrt{\frac{2}{\sigma_{w}^2}}\tilde{\textbf{\emph{a}}}_{R,p_{\mathrm{\uppercase\expandafter{\romannumeral1}}},i}^T\acute{\mathbf{h}}_{d_{R}})).
			\\
		\end{split}
	\end{equation}	
The CRLB of the direct channel estimation is presented in the following theorem.
	\begin{theorem}
		In $\mathrm{1}$-bit quantization system, if $\breve{\mathbf{h}}_{d_{R}}$ is the unbiased estimation of real direct channel $\mathbf{h}_{d_{R}}$, the mean square error of channel estimation is lower bounded by
			\begin{equation}\label{System-28}
				\begin{split}
					\mathrm{MSE}(\breve{\mathbf{h}}_{d_{R,ML}})\geq\mathrm{tr}(\mathbf{J}^{-1}),\\	
				\end{split}
			\end{equation}
where ${\mathbf J} \in {\mathbb C}^{2M\times 2M}$ is Fisher information matrix given as
			\begin{equation}\label{System-29}
				\begin{split}
					\mathbf{J}&=
					\overset{2M\tau}{\underset{i=1}{\sum}}
					\frac{\frac{2}{\sigma_{w}^2}(\Phi^{'}(\sqrt{\frac{2}{\sigma_{w}^2}}\textbf{\emph{a}}^T_{R,p_{\mathrm{\uppercase\expandafter{\romannumeral1}}},i}\mathbf{h}_{d_{R}}))^2\tilde{\textbf{\emph{a}}}_{R,p_{\mathrm{\uppercase\expandafter{\romannumeral1}}},i}\tilde{\textbf{\emph{a}}}^T_{R,p_{\mathrm{\uppercase\expandafter{\romannumeral1}}},i}}
					{\Phi(\sqrt{\frac{2}{\sigma_{w}^2}}\textbf{\emph{a}}^T_{R,p_{\mathrm{\uppercase\expandafter{\romannumeral1}}},i}\mathbf{h}_{d_{R}})(1-\Phi(\sqrt{\frac{2}{\sigma_{w}^2}}\textbf{\emph{a}}^T_{R,p_{\mathrm{\uppercase\expandafter{\romannumeral1}}},i}\mathbf{h}_{d_{R}}))}
					\\	
 				\end{split}
			\end{equation}
		with
			\begin{equation}\label{System-30}
					\Phi^{'}(\mathrm{x})=\frac{1}{\sqrt{2\pi  }}\mathrm{exp}(-\frac{\mathrm{x}^2}{2}).
			\end{equation}
	\end{theorem}
	\begin{proof}
		By taking the first derivative and the second derivative of ML function \eqref{System-17}, we can have 
			\begin{equation}\label{System-ML-log1}	
				\begin{split}
					L^{'}(\mathbf{h}_{d_{R}})=
					\overset{2M\tau}{\underset{i=1}{\sum}}
					\frac{\sqrt{\frac{2}{\sigma_{w}^2}}\Phi^{'}(\sqrt{\frac{2}{\sigma_{w}^2}}\tilde{\textbf{\emph{a}}}^T_{R,p_{\mathrm{\uppercase\expandafter{\romannumeral1}}},i}\mathbf{h}_{d_{R}})}
					{\Phi(\sqrt{\frac{2}{\sigma_{w}^2}}\tilde{\textbf{\emph{a}}}^T_{R,p_{\mathrm{\uppercase\expandafter{\romannumeral1}}},i}\mathbf{h}_{d_{R}})}
					\tilde{\textbf{\emph{a}}}_{R,p_{\mathrm{\uppercase\expandafter{\romannumeral1}}},i},
				\end{split}
			\end{equation}
and \eqref{System-ML-log2}.
			\newcounter{TempEqCnt} 
			\setcounter{TempEqCnt}{\value{equation}} 
			\setcounter{equation}{29} 
			\begin{figure*}[ht] 
				\begin{equation}\label{System-ML-log2}	
				\begin{split}
					L^{''}(\mathbf{h}_{d_{R}})=
					\overset{2M\tau}{\underset{i=1}{\sum}}
					\left( \frac{(2/\sigma_{w}^2)\Phi^{''}(\sqrt{(2/\sigma_{w}^2)}\tilde{\textbf{\emph{a}}}^T_{R,p_{\mathrm{\uppercase\expandafter{\romannumeral1}}},i}\mathbf{h}_{d_{R}})}
					{\Phi(\sqrt{(2/\sigma_{w}^2)}\tilde{\textbf{\emph{a}}}^T_{R,p_{\mathrm{\uppercase\expandafter{\romannumeral1}}},i}\mathbf{h}_{d_{R}})}
					-
					\frac{(2/\sigma_{w}^2)(\Phi^{'}(\sqrt{(2/\sigma_{w}^2)}\tilde{\textbf{\emph{a}}}^T_{R,p_{\mathrm{\uppercase\expandafter{\romannumeral1}}},i}\mathbf{h}_{d_{R}}))^2}
					{\Phi^2(\sqrt{(2/\sigma_{w}^2)}\tilde{\textbf{\emph{a}}}^T_{R,p_{\mathrm{\uppercase\expandafter{\romannumeral1}}},i}\mathbf{h}_{d_{R}})}
					 \right)
					\tilde{\textbf{\emph{a}}}_{R,p_{\mathrm{\uppercase\expandafter{\romannumeral1}}},i}\tilde{\textbf{\emph{a}}}^T_{R,p_{\mathrm{\uppercase\expandafter{\romannumeral1}}},i}.
				\end{split}
			\end{equation}	
					
			\begin{equation}\label{System-Fisher}	
				\begin{split}
					\mathbf{J}=&-\mathrm{E}[{\mathrm{L}^{''}(\mathbf{h}_{d_{R}})}]		=-\overset{2M\tau}{\underset{i=1}{\sum}}
					\bigg(
					\mathrm{Pr}(\textbf{r}_{R,{p_{\mathrm{\uppercase\expandafter{\romannumeral1}}}},i}=1)
					(\frac{(2/\sigma_{w}^2)\Phi^{''}(\sqrt{(2/\sigma_{w}^2)}\textbf{\emph{a}}^T_{R,p_{\mathrm{\uppercase\expandafter{\romannumeral1}}},i}\mathbf{h}_{d_{R}})}
					{\Phi(\sqrt{(2/\sigma_{w}^2)}\textbf{\emph{a}}^T_{R,p_{\mathrm{\uppercase\expandafter{\romannumeral1}}},i}\mathbf{h}_{d_{R}})}
					-
					\frac{(2/\sigma_{w}^2)(\Phi^{'}(\sqrt{(2/\sigma_{w}^2)}\textbf{\emph{a}}^T_{R,p_{\mathrm{\uppercase\expandafter{\romannumeral1}}},i}\mathbf{h}_{d_{R}}))^2}
					{\Phi^2(\sqrt{(2/\sigma_{w}^2)}\textbf{\emph{a}}^T_{R,p_{\mathrm{\uppercase\expandafter{\romannumeral1}}},i}\mathbf{h}_{d_{R}})})\\
					+&\mathrm{Pr}(\textbf{r}_{R,{p_{\mathrm{\uppercase\expandafter{\romannumeral1}}}},i}=-1)
					(\frac{(2/\sigma_{w}^2)\Phi^{''}(-\sqrt{(2/\sigma_{w}^2)}\textbf{\emph{a}}^T_{R,p_{\mathrm{\uppercase\expandafter{\romannumeral1}}},i}\mathbf{h}_{d_{R}})}
					{\Phi(-\sqrt{(2/\sigma_{w}^2)}\textbf{\emph{a}}^T_{R,p_{\mathrm{\uppercase\expandafter{\romannumeral1}}},i}\mathbf{h}_{d_{R}})}
					-
					\frac{(2/\sigma_{w}^2)(\Phi^{'}(-\sqrt{(2/\sigma_{w}^2)}\textbf{\emph{a}}^T_{R,p_{\mathrm{\uppercase\expandafter{\romannumeral1}}},i}\mathbf{h}_{d_{R}}))^2}
					{\Phi^2(-\sqrt{(2/\sigma_{w}^2)}\textbf{\emph{a}}^T_{R,p_{\mathrm{\uppercase\expandafter{\romannumeral1}}},i}\mathbf{h}_{d_{R}})})											
					\bigg)\tilde{\textbf{\emph{a}}}_{R,p_{\mathrm{\uppercase\expandafter{\romannumeral1}}},i}\tilde{\textbf{\emph{a}}}^T_{R,p_{\mathrm{\uppercase\expandafter{\romannumeral1}}},i} \\	
					=&\overset{2M\tau}{\underset{i=1}{\sum}}
					\bigg(
					\frac{(2/\sigma_{w}^2)(\Phi^{'}(\sqrt{(2/\sigma_{w}^2)}\textbf{\emph{a}}^T_{R,p_{\mathrm{\uppercase\expandafter{\romannumeral1}}},i}\mathbf{h}_{d_{R}}))^2}
					{\Phi(\sqrt{(2/\sigma_{w}^2)}\textbf{\emph{a}}^T_{R,p_{\mathrm{\uppercase\expandafter{\romannumeral1}}},i}\mathbf{h}_{d_{R}})}+
					\frac{(2/\sigma_{w}^2)(\Phi^{'}(-\sqrt{(2/\sigma_{w}^2)}\textbf{\emph{a}}^T_{R,p_{\mathrm{\uppercase\expandafter{\romannumeral1}}},i}\mathbf{h}_{d_{R}}))^2}
					{\Phi(-\sqrt{(2/\sigma_{w}^2)}\textbf{\emph{a}}^T_{R,p_{\mathrm{\uppercase\expandafter{\romannumeral1}}},i}\mathbf{h}_{d_{R}})}
					\bigg)\tilde{\textbf{\emph{a}}}_{R,p_{\mathrm{\uppercase\expandafter{\romannumeral1}}},i}\tilde{\textbf{\emph{a}}}^T_{R,p_{\mathrm{\uppercase\expandafter{\romannumeral1}}},i}.	
			\end{split}
			\end{equation}			
			\hrulefill  
		\end{figure*}
Based on \eqref{System-ML-log2}, Fisher information matrix can be written as \eqref{System-Fisher}. According to the property of cumulative distribution function for a normal distribution, we obtain the Fisher information matrix in \eqref{System-29}.
	\end{proof}
Assuming that the estimation of direct channel is unbiased, the error between the estimated channel and the actual channel obeys the Gaussian distribution, i.e., $\mathbf{h}_{d}-\breve{\textbf{\emph{h}}}_{d_{ML}}=\textbf{e}_{d}$, the elements of the error are drown from $\mathcal{CN}(0,\sigma_e^2/M)$ independently, with
	\begin{equation}
		\begin{split}				
			\sigma_e^2&=\sqrt{2}{tr}(\mathbf{J}^{-1}).
		\end{split}
	\end{equation}

\subsection{Reflecting Channel estimation}
In Phase $\mathrm{\uppercase\expandafter{\romannumeral2}}$, user sends a pilot sequence consisting of $\tau$ symbols given as
	\begin{equation}\label{System-addd1}
		\begin{split}
			\textbf{\emph{a}}_\mathrm{\uppercase\expandafter{\romannumeral2}} =[a_1,\ldots,a_\tau]^T,
		\end{split}
	\end{equation}
and the received signal at the BS in Phase $\mathrm{\uppercase\expandafter{\romannumeral2}}$ is
	\begin{equation}\label{System-20}
			\textbf{Y}_{p_{\mathrm{\uppercase\expandafter{\romannumeral2}}}}			=(\mathbf{G}\pmb\Theta\mathbf{h}_{r}+\mathbf{h}_{d})\textbf{\emph{a}}_\mathrm{\uppercase\expandafter{\romannumeral2}}^T+\mathbf{W}_{p_{\mathrm{\uppercase\expandafter{\romannumeral2}}}}=(\textbf{H}_{p_{\mathrm{\uppercase\expandafter{\romannumeral2}}}}\pmb\theta+\mathbf{h}_{d})\textbf{\emph{a}}_\mathrm{\uppercase\expandafter{\romannumeral2}}^T+\mathbf{W}_{p_{\mathrm{\uppercase\expandafter{\romannumeral2}}}},
	\end{equation}
  where $\textbf{H}_{p_{\mathrm{\uppercase\expandafter{\romannumeral2}}}}=\textbf{G}\mathrm{diag}(\mathbf{h}_{r})$. After $\mathrm{1}$-bit quantization, the signal can be expressed as $\textbf{r}_{p_{\mathrm{\uppercase\expandafter{\romannumeral2}}}}=\mathcal{Q}(\mathrm{vec}(\textbf{Y}_{p_{\mathrm{\uppercase\expandafter{\romannumeral2}}}}))$, and $\textbf{r}_{R,p_{\mathrm{\uppercase\expandafter{\romannumeral2}}}}$ represents the vector after realizable operation.
	
The signal at the BS through the reflecting channel is
	\begin{equation}\label{System-222}
			\textbf{Y}^{ref}_{p_{\mathrm{\uppercase\expandafter{\romannumeral2}}}} =\textbf{H}_{p_{\mathrm{\uppercase\expandafter{\romannumeral2}}}}\pmb\theta\textbf{\emph{a}}_\mathrm{\uppercase\expandafter{\romannumeral2}}^T+\textbf{e}_d\textbf{\emph{a}}_\mathrm{\uppercase\expandafter{\romannumeral2}}^T+\mathbf{W}_{p_{\mathrm{\uppercase\expandafter{\romannumeral2}}}}.
	\end{equation}
We aim to estimate $\textbf{H}_{p_{\mathrm{\uppercase\expandafter{\romannumeral2}}}}$.
	The received signal in \eqref{System-222} after vectorization is
	\begin{equation}\label{System-21}	\mathbf{y}_{p_{\mathrm{\uppercase\expandafter{\romannumeral2}}}}^{ref}=\mathbf{A}_{p_{\mathrm{\uppercase\expandafter{\romannumeral2}}}}\textbf{h}_{p_{\mathrm{\uppercase\expandafter{\romannumeral2}}}}+\bm{\varrho}_{p_{\mathrm{\uppercase\expandafter{\romannumeral2}}}}+\mathbf{w}_{p_{\mathrm{\uppercase\expandafter{\romannumeral2}}}},
	\end{equation}
	where $\mathbf{A}_{p_{\mathrm{\uppercase\expandafter{\romannumeral2}}}}=\textbf{\emph{a}}_\mathrm{\uppercase\expandafter{\romannumeral2}}\pmb\theta^T\otimes\mathbf{I}_{M}\in\mathbb{C}^{M\tau\times MN}$, $\bm{\varrho}_{p_{\mathrm{\uppercase\expandafter{\romannumeral2}}}}=\textbf{\emph{a}}_\mathrm{\uppercase\expandafter{\romannumeral2}}\otimes\textbf{e}_d$, and $\textbf{e}_d$ and $\text{\textbf{w}}_{p_{\mathrm{\uppercase\expandafter{\romannumeral2}}}}$.
	Similarly to the previous section, we reformulate all expressions into real domain as
	\begin{equation}\label{System-23}
		\mathbf{y}_{R,p_{\mathrm{\uppercase\expandafter{\romannumeral2}}}}^{ref}=	\mathbf{A}_{R,p_{\mathrm{\uppercase\expandafter{\romannumeral2}}}}\textbf{h}_{R,p_{\mathrm{\uppercase\expandafter{\romannumeral2}}}}+\bm{\varrho}_{R,p_{\mathrm{\uppercase\expandafter{\romannumeral2}}}}+\mathbf{w}_{R,p_{\mathrm{\uppercase\expandafter{\romannumeral2}}}},
	\end{equation}
	where	
	\begin{align}		
		\mathbf{y}_{R,p_{\mathrm{\uppercase\expandafter{\romannumeral2}}}}^{ref}&=\begin{bmatrix}
											\mathrm{Re}(\mathbf{y}^T_{p_{\mathrm{\uppercase\expandafter{\romannumeral2}}}}) & \mathrm{Im}(\mathbf{y}^T_{p_{\mathrm{\uppercase\expandafter{\romannumeral2}}}})
										\end{bmatrix}^T\in\mathbb{R}^{2M\tau\times 1},
		\\
		\mathbf{A}_{R,p_{\mathrm{\uppercase\expandafter{\romannumeral2}}}}&=\begin{bmatrix}
										\mathrm{Re}(\mathbf{A}_{p_{\mathrm{\uppercase\expandafter{\romannumeral2}}}})&  -\mathrm{Im}(\mathbf{A}_{p_{\mathrm{\uppercase\expandafter{\romannumeral2}}}})  \\ \mathrm{Im}( \mathbf{A}_{p_{\mathrm{\uppercase\expandafter{\romannumeral2}}}}) & \mathrm{Re}(\mathbf{A}_{p_{\mathrm{\uppercase\expandafter{\romannumeral2}}}})
										\end{bmatrix}\in\mathbb{R}^{2M\tau\times 2MN},
		\\
		\textbf{h}_{R,p_{\mathrm{\uppercase\expandafter{\romannumeral2}}}}&=\begin{bmatrix}
										\mathrm{Re}(\textbf{h}^T_{p_{\mathrm{\uppercase\expandafter{\romannumeral2}}}}) &
										\mathrm{Im}(\textbf{h}^T_{p_{\mathrm{\uppercase\expandafter{\romannumeral2}}}})
										\end{bmatrix}^T\in\mathbb{R}^{2MN\times 1},
		\\
			\bm{\varrho}_{R,p_{\mathrm{\uppercase\expandafter{\romannumeral2}}}}&=\begin{bmatrix}
										\mathrm{Re}(\bm{\varrho}^T_{p_{\mathrm{\uppercase\expandafter{\romannumeral2}}}}) &
										\mathrm{Im}(\bm{\varrho}^T_{p_{\mathrm{\uppercase\expandafter{\romannumeral2}}}})
								     \end{bmatrix}^T\in\mathbb{R}^{2M\tau\times 1},		
		\\
		\mathbf{w}_{R,p_{\mathrm{\uppercase\expandafter{\romannumeral2}}}}&=\begin{bmatrix}
										\mathrm{Re}(\mathbf{w}^T_{p_{\mathrm{\uppercase\expandafter{\romannumeral2}}}}) &
										\mathrm{Im}(\mathbf{w}^T_{p_{\mathrm{\uppercase\expandafter{\romannumeral2}}}})
									 \end{bmatrix}^T\in\mathbb{R}^{2M\tau\times 1}.
	\end{align}
Note that the output of the $\mathrm{1}$-bit ADCs after vectorization is $\textbf{r}_{p_{\mathrm{\uppercase\expandafter{\romannumeral2}}}}$, and $\textbf{r}_{R,p_{\mathrm{\uppercase\expandafter{\romannumeral2}}}}$ represents its real form. The BS performs the sign-refinement as
	\begin{equation}\label{System-sign-refinement}	
			\tilde{\textbf{\emph{a}}}_{R,p_{\mathrm{\uppercase\expandafter{\romannumeral2}}},i} =r_{R,p_{\mathrm{\uppercase\expandafter{\romannumeral2}}},i}\textbf{\emph{a}}_{R,p_{\mathrm{\uppercase\expandafter{\romannumeral2}}},i},
	\end{equation}	
where $\textbf{\emph{a}}_{R,p_{\mathrm{\uppercase\expandafter{\romannumeral2}}},i}^T$ is the $i$-th line of $\mathbf{A}_{R,p_{\mathrm{\uppercase\expandafter{\romannumeral2}}},i}$, $r_{R,p_{\mathrm{\uppercase\expandafter{\romannumeral2}}},i}$ is the $i$-th element of $\textbf{r}_{R,p_{\mathrm{\uppercase\expandafter{\romannumeral2}}}}$. We define two sets $\mathcal{S}$ and $\mathcal{P}$ as
	\begin{equation}\label{System-26}	
			\mathcal{S}=\{i:r_{R,p_{\mathrm{\uppercase\expandafter{\romannumeral2}}},i}\geq 0\},	\mathcal{P}=\{i:r_{R,p_{\mathrm{\uppercase\expandafter{\romannumeral2}}},i}< 0\}.
	\end{equation}
With these definitions, the likelihood function is
	\begin{subequations}
		\begin{align}
				L(\acute{\textbf{h}}_{R,p_{\mathrm{\uppercase\expandafter{\romannumeral2}}}})\notag=&			
					Pr(\frac{\tilde{\textbf{\emph{a}}}_{R,p_{\mathrm{\uppercase\expandafter{\romannumeral2}}},i}^T\acute{\textbf{h}}_{R,p_{\mathrm{\uppercase\expandafter{\romannumeral2}}}}}{\sqrt{\frac{\frac{\sigma_{e}^2}{M}||\textbf{\emph{a}}_{\mathrm{\uppercase\expandafter{\romannumeral2}},\lceil\frac{i}{M}\rceil\% \tau}||^2+\sigma_{w}^2}{2}}}
					\geq-\textbf{\textit{w}}_{R,p_{\mathrm{\uppercase\expandafter{\romannumeral2}}},i}|\forall{i}\in\mathcal{S})\notag\\				
					&\cdot Pr(\frac{\tilde{\textbf{\emph{a}}}_{R,p_{\mathrm{\uppercase\expandafter{\romannumeral2}}},i}^T\acute{\textbf{h}}_{R,p_{\mathrm{\uppercase\expandafter{\romannumeral2}}}}}{\sqrt{\frac{\frac{\sigma_{e}^2}{M}||\textbf{\emph{a}}_{\mathrm{\uppercase\expandafter{\romannumeral2}},\lceil\frac{i}{M}\rceil\% \tau}||^2+\sigma_{w}^2}{2}}}\geq\textbf{\textit{w}}_{R,p_{\mathrm{\uppercase\expandafter{\romannumeral2}}},i}|\forall{i}\in\mathcal{P})\label{Eqn_d}\\
					=&
					\overset{2M\tau}{\underset{i=1}{\prod}}
					\Phi(\frac{\tilde{\textbf{\emph{a}}}_{R,p_{\mathrm{\uppercase\expandafter{\romannumeral2}}},i}^T\acute{\textbf{h}}_{R,p_{\mathrm{\uppercase\expandafter{\romannumeral2}}}}}{\sqrt{(\frac{\sigma_{e}^2}{M}||\textbf{\emph{a}}_{\mathrm{\uppercase\expandafter{\romannumeral2}},\lceil\frac{i}{M}\rceil\% \tau}||^2+\sigma_{w}^2)/2}})\label{Eqn_e},
		\end{align}
	\end{subequations}	
where $\Phi(x)$ is cumulative density function, i.e., $\Phi(x)=\int_{-\infty}^x\frac{1}{\sqrt{2\pi}}e^{-\frac{t^2}{2}}dt$, and $\textbf{\emph{a}}_{\mathrm{\uppercase\expandafter{\romannumeral2}},\lceil\frac{i}{M}\rceil\% \tau}$ is the $\lceil\frac{i}{M}\rceil\% \tau$-th element of $\textbf{\emph{a}}_\mathrm{\uppercase\expandafter{\romannumeral2}}$. Eq. \eqref{Eqn_e} is obtained due to the fact that $\text{\emph{w}}_{R,m,i}$ and $-\text{\emph{w}}_{R,m,i}$ have the same probability density function, i.e., $\mathrm{Pr}(k\geq\text{\emph{w}}_{R,m,i})=\mathrm{Pr}(k\geq-\text{\emph{w}}_{R,m,i})$, for an arbitrary constant $k$. Then the ML estimator is given as
	\begin{equation}\label{System-reflectChannel}	
		\begin{split}
			\breve{\textbf{h}}_{R,p_{\mathrm{\uppercase\expandafter{\romannumeral2}}},ML}=
			\underset{\acute{\textbf{h}}_{R,p_{\mathrm{\uppercase\expandafter{\romannumeral2}}}}\in\mathbb{R}^{2M\times 1}}{\mathrm{argmax}}\overset{2M\tau}{\underset{i=1}{\sum}}						\mathrm{ln}\bigg(\Phi(\frac{\tilde{\textbf{\emph{a}}}_{R,p_{\mathrm{\uppercase\expandafter{\romannumeral2}}},i}^T\acute{\textbf{h}}_{R,p_{\mathrm{\uppercase\expandafter{\romannumeral2}}}}}{\sqrt{\frac{\frac{\sigma_{e}^2}{M}||\textbf{\emph{a}}_{\mathrm{\uppercase\expandafter{\romannumeral2}},\lceil\frac{i}{M}\rceil\% \tau}||^2+\sigma_{w}^2}{2}}})\bigg),
		\end{split}
	\end{equation}
	which can be solved by using the method presented in Algorithm 1.
	
	 \begin{algorithm}
	 		\label{alg:aa}
	 		\caption{: Iterative algorithm to solve \eqref{System-reflectChannel}} 
	 		\hspace*{0.02in} {\bf Initialize:} 
	 		
 			\begin{algorithmic}[1]
 				\State Set the initial point $\acute{\textbf{h}}_{R,p_{\mathrm{\uppercase\expandafter{\romannumeral2}}},ML}^{(0)}$ 
 				\State Set current iteration number $n=0$, and the maximum number of iterations is $\mathrm{N_{ITRT}}$
 				\State Set the termination threshold $\epsilon$	and step size $\alpha$
			
		    \hspace*{-0.5in} {\bf Iterative update:} 
 				\While{$||\acute{\textbf{h}}_{R,p_{\mathrm{\uppercase\expandafter{\romannumeral2}}},ML}^{(k)}-\acute{\textbf{h}}_{R,p_{\mathrm{\uppercase\expandafter{\romannumeral2}}},ML}^{(k-1)}||\geq\epsilon||\acute{\textbf{h}}_{R,p_{\mathrm{\uppercase\expandafter{\romannumeral2}}},ML}^{(k-1)}||$}
	 				\State $\acute{\textbf{h}}_{R,p_{\mathrm{\uppercase\expandafter{\romannumeral2}}},ML}^{(k)}=\acute{\textbf{h}}_{R,p_{\mathrm{\uppercase\expandafter{\romannumeral2}}},ML}^{(k-1)}+\alpha \nabla\mathit{v}(\acute{\textbf{h}}_{R,p_{\mathrm{\uppercase\expandafter{\romannumeral2}}},ML}^{(k-1)})$

	 				where
	 				
	 				$\nabla\mathit{v}(\acute{\textbf{h}}_{R,p_{\mathrm{\uppercase\expandafter{\romannumeral2}}},ML}^{(k-1)})=				
	 				\frac{1}{\sqrt{\pi}}\overset{2M\tau}{\underset{i=1}{\sum}}	
	 				\frac{1}{\sqrt{x_{den} } } 				\frac{\mathit{e}^{-\frac{||x_{num}||^2}{x_{den} }}}			{\Phi\big(\frac{x_{num}}{\sqrt{\frac{x_{den} }{2}}}\big)}\tilde{\textbf{\emph{a}}}_{R,p_{\mathrm{\uppercase\expandafter{\romannumeral2}}},i}$
	 				
	 				$x_{num} = \tilde{\textbf{\emph{a}}}_{R,p_{\mathrm{\uppercase\expandafter{\romannumeral2}}},i}^T\acute{\textbf{h}}_{R,p_{\mathrm{\uppercase\expandafter{\romannumeral2}}}}$
	 				
	 				$x_{den} = \frac{\sigma_{e}^2}{M}||\textbf{\emph{a}}_{\mathrm{\uppercase\expandafter{\romannumeral2}},\lceil\frac{i}{M}\rceil\% \tau}||^2+\sigma_{w}^2$	
 				\EndWhile

 			\end{algorithmic} 					
	 	\end{algorithm}

\section{simulation results and analysis}
In this section we present Monte-Carlo simulations of passive beamforming design and channel estimation to evaluate the proposed techniques. We set the amplitude coefficients of all elements in IRS to be $\mathrm{1}$ and the transmission power at the user is $\mathrm{P=1}$dB, so the SNR is defined as $\mathrm{SNR}=1/{\sigma_{w}^2}$. 

\subsection{Passive beamforming design}
Fig. \ref{Fig_2} illustrates the achievable rate as a function of SNR with different elements in IRS under $\mathrm{1}$-bit quantization. It is found that the system with IRS significantly outperforms the one without IRS, which proves that the IRS can effectively increase the achievable rate. The SDP method, gradient decent (GD) and BB method proposed in this paper achieves a better system performance compared to phase matching (PM). SDP method and GD method both overlap with BB, implying that they can achieve the optimal value. The result shows that achievable rate reaches a floor at high SNR regime due to the effect of $\mathrm{1}$-bit quantization at the receiver. In Fig. \ref{Fig_3}, we plot curves of achievable rate versus SNR with 10-bits quantization. It is noticed that
the proposed GD method performs better than the PM method. From Fig. \ref{Fig_2} and Fig. \ref{Fig_3}, we observe that the achievable rate depends on the bits of quantization. In low SNR regime, the achievable rate can be efficiently improved through
the proposed IRS optimization, while the achievable rate is bounded at the high SNR by the number quantization bits.
	\begin{figure}[htbp]
	\centering
		\vspace{-0.1cm}
	\includegraphics[scale=0.45]{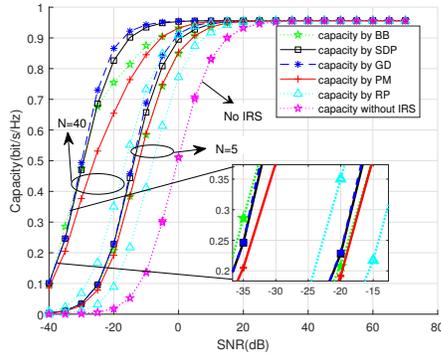}
	\vspace{-0.4cm}
	\caption{Achievable rate with 1-bit quantization and N=5 vs N=40.}
	\label{Fig_2}
	\centering
	\vspace{-0.1cm}
	\end{figure}

	\begin{figure}[htbp]
		\centering
			\vspace{-0.1cm}
		\includegraphics[scale=0.45]{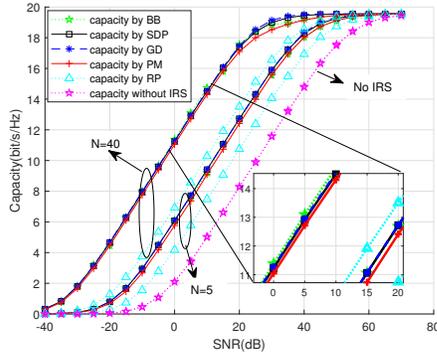}
		\vspace{-0.4cm}
		\caption{Achievable rate with 10-bits quantization and N=5 vs N=40.} \label{Fig_3}
		\centering
		\vspace{-0.3cm}
	\end{figure}

\subsection{Reflecting Channel estimation}

In this subsection, we evaluate the proposed channel estimation with $\mathrm{1}$-bit quantization, and compare the performance of three estimators, i.e., the proposed ML, least squares (LS) and linear minimum mean square error estimation (LMMSE). For LMMSE and LS estimator, we estimate the direct channel first, then turn on the IRS to estimate the reflecting channel in a coherence time.

In Fig. \ref{Fig_4}, we evaluate the performance of the estimators with different training length $\tau$ and SNR values. It can be seen that in low SNR, the performance of all three estimators improves as the pilot length increases, ML and LMMSE methods performed significantly better than LS method, while the performance of all three estimators become close and saturates as the SNR increases. This implies that we can't improve the performance of estimators by increasing the pilot length in high SNR due to 1-bit quantization. Similar observations can also be found in Fig. \ref{Fig_5}. The curves in Figs. \ref{Fig_4} and \ref{Fig_5} demonstrate that the normalized MSE (NMSE) decreases when the number of elements in IRS increase in low SNR, which suggests that IRS has a positive effect on channel estimation. Although both the ML and LMMSE methods can minimize the NMSE, the LMMSE method is impractical as it only estimates one column of the channel at a time slot $\tau$, and requires the channel statistics information. 


\begin{figure}[htbp]
	\centering
		\vspace{-0.3cm}
	\includegraphics[scale=0.455]{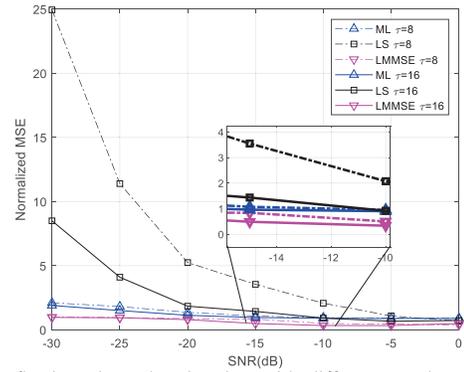}
	\vspace{-0.4cm}
	\caption{Reflecting channel estimation with different $\tau$ when $M=2$ and $N=4$.}
	\label{Fig_4}
	\centering
	\vspace{-0.1cm}
	\end{figure}

\begin{figure}[htbp]
	\centering
		\vspace{-0.3cm}
	\includegraphics[scale=0.45]{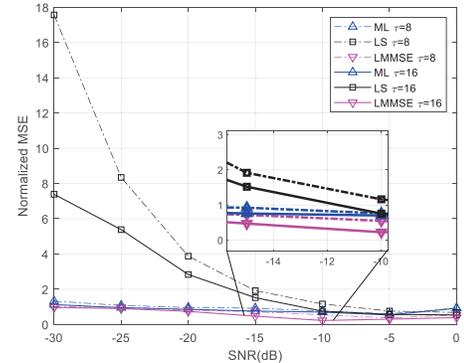}
	\vspace{-0.4cm}
	\caption{Reflecting channel estimation with different $\tau$ when $M=2$ and $N=8$.}
	\label{Fig_5}
	\centering
	\vspace{-0.1cm}
	\end{figure}

\section{Conclusion}
In this paper, we studied the beamforming optimization and channel estimation
for SIMO communication system with few-bit ADCs. The achievable rate was maximized by applying the different optimization techniques. Moreover, we proposed to use the ML estimator to obtain more accurate reflecting channel estimation in $\mathrm{1}$-bit quantization scenario.

\bibliographystyle{IEEEtran}
\bibliography{paper_ref}
\end{document}